\documentclass[aos,preprint]{imsart}%
\RequirePackage[T1]{fontenc}%
\RequirePackage{amsthm,amsmath}%
\RequirePackage[authoryear]{natbib}% or authoryear
\RequirePackage[citecolor=blue,urlcolor=blue,breaklinks]{hyperref}%
\usepackage{amsfonts, amssymb, graphicx, mathtools, epstopdf}%
\usepackage{accents,bm}%
\usepackage[caption=false]{subfig}%
\theoremstyle{plain}%
	\newtheorem{corollary}{Corollary}%
	\newtheorem{theorem}{Theorem}%
\theoremstyle{definition}%
	\newtheorem{assumption}{Assumption}%
	\newtheorem{definition}{Definition}%
\newcommand{\A}{\mathcal{A}}%
\newcommand{\E}{\mathbb{E}}%
\newcommand{\cov}{\mathrm{cov}}%
\newcommand{\var}{\mathrm{var}}%
\let\originalleft\left%
\let\originalright\right%
\renewcommand{\left}{\mathopen{}\mathclose\bgroup\originalleft}%
\renewcommand{\right}{\aftergroup\egroup\originalright}%
\makeatletter%
\def\mybigx#1{\dimen@#1\relax%
\mathchoice%
{\vbox to \dimen@{}}%
{\vbox to \dimen@{}}%
{\vbox to .7\dimen@{}}%
{\vbox to .5\dimen@{}}}%
\def\mybig#1{{\hbox{$\left#1\mybigx{0.8em}\right.\n@space$}}}%
\makeatother%
\begin{document}%%%%%%%%%%%%%%%%%%%%%%%%
\begin{frontmatter}%%%%%%%%%%%%%%%%%%%%%%%%
\title{Event conditional correlation}
\runtitle{Event conditional correlation}
\begin{aug}
\author{\fnms{Pierre-Andr\'e G.} \snm{Maugis}\ead[label=e1]{p.maugis@ucl.ac.uk}}
\affiliation{University College London}
\address{Department of Statistical Science\\
University College London\\
Gower Street\\
London WC1E 6BT, UK\\
\printead{e1}}
\runauthor{P-A. G. Maugis}
\end{aug}
\begin{abstract}%%%%%%%%%%%%%%%%%%%%%%%%%
Entries of datasets are often collected only if an event occurred: taking a survey, enrolling in an experiment and so forth. However, such partial samples bias classical correlation estimators. Here we show how to correct for such sampling effects through two complementary estimators of event conditional correlation: the correlation of two random variables conditional on a given event. First, we provide under minimal assumptions proof of consistency and asymptotic normality for the proposed estimators. Then, through synthetic examples, we show that these estimators behave well in small-sample and yield powerful methodologies for non-linear regression as well as dependence testing. Finally, by using the two estimators in tandem, we explore counterfactual dependence regimes in a financial dataset. By so doing we show that the contagion which took place during the 2007--2011 financial crisis cannot be explained solely by increased financial risk.
\end{abstract}%%%%%%%%%%%%%%%%%%%%%%%%%%
\end{frontmatter}%%%%%%%%%%%%%%%%%%%%%%%%%
\section{Introduction}%%%%%%%%%%%%%%%%%%%%%%%
We provide methods to estimate and compare correlation estimates based on partial samples. We do so by deriving under minimal assumptions the properties of a new dependence parameter we introduce: event conditional correlation.

We define event conditional correlation as the correlation of two variables $X$ and $Y$ conditionally to an event $\A$ and denote it $\rho_{XY\mid \A}$. Event conditional correlation is the natural correlation parameter when working with partial samples. Consider the case where one is able to measure $(X,Y)$ only if a third random variable $Z$ is large enough, say larger than a threshold $z$. One classical example~\cite{Akemann83} is knowing the grades of students (the variables $X$ and $Y$), only if these students had high enough scores in high school (the variable $Z$) to enter university. In this setting, naively using the ordinary least squares estimator of correlation on the available sample produces an estimate of $\rho_{XY\mid Z>z}$. Such an estimate can be sensibly different from $\rho_{XY}$, the classical or unconditional correlation parameter. We provide a quantified example in Fig~\ref{fig:GCorrel} where $(X,Y,Z)$ is a trivariate Gaussian vector.

\begin{figure}
\centering
  \subfloat[Correlation across quantiles]{%
  \label{fig:GCorrel}%
  \includegraphics[width=.4\textwidth,height=.4\textwidth]{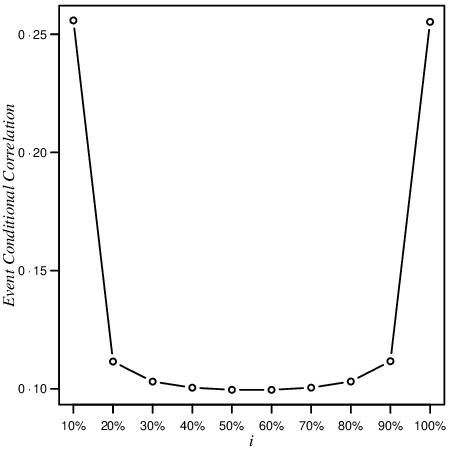}}
  \subfloat[Eigenvectors across quantiles]{%
  \label{fig:PCA}%
  \includegraphics[width=.445\textwidth,height=.4\textwidth,
  	trim = 3.75cm 3.75cm 3.75cm 3.75cm,clip]{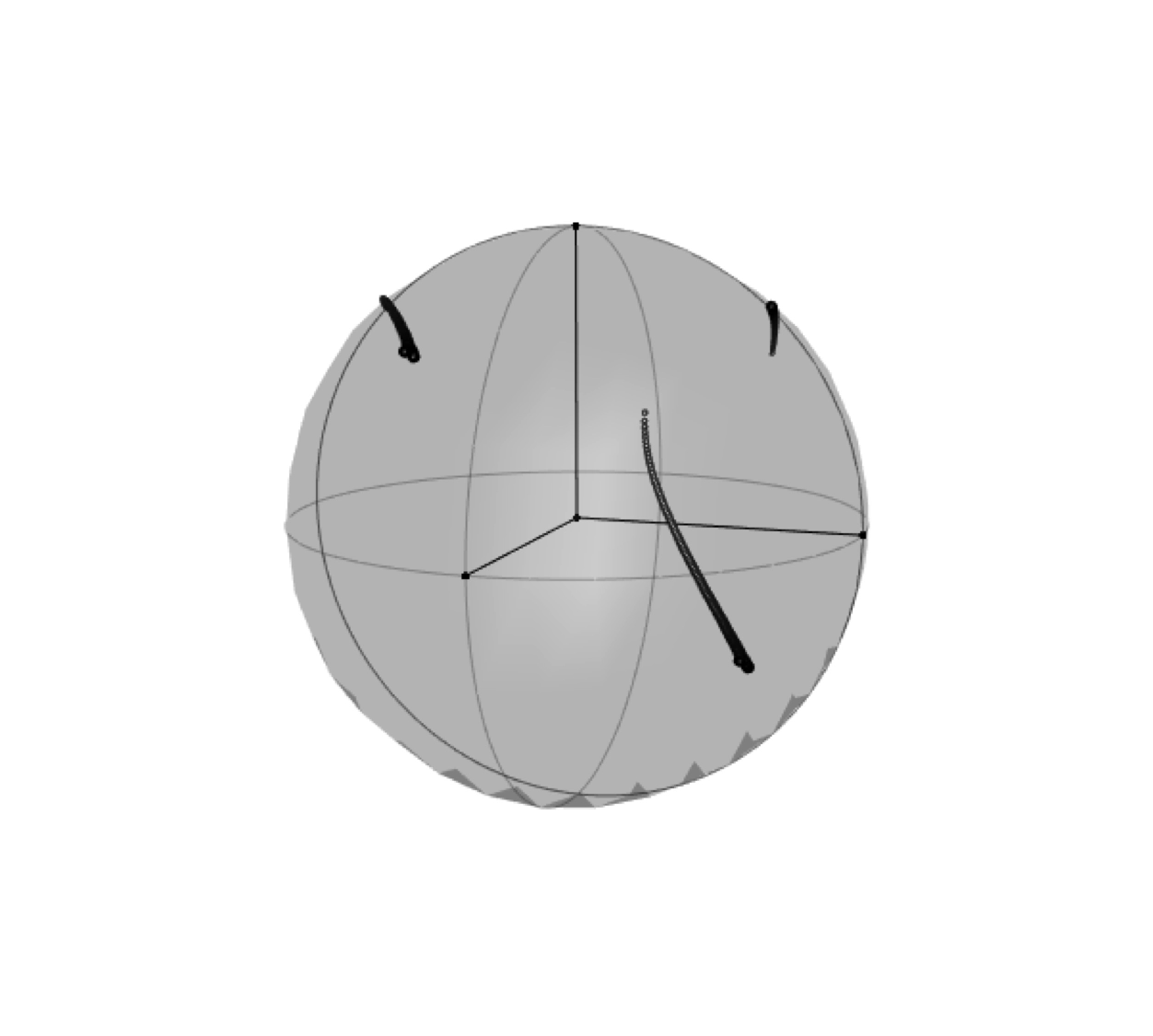}}
  \caption{%
  Effect of sampling on dependence structure. %
  (a) Plot of $\rho_{XY\mid Z \in [Q_Z(i-0\cdot1),Q_Z(i)]}$ as a function of $i$ (in percent) where $Q_Z$ is the quantile function of $Z$ and $(X,Y,Z)$ follows a trivariate Gaussian distribution (with $\rho_{XY} = $ 0$\,\cdot$6, $\rho_{XZ} = $ 0$\,\cdot$7 and $\rho_{YZ} = $ 0$\,\cdot$8). %
  (b) Plot of the eigenvectors of the covariance matrix of $(X_1,X_2,X_3\mid Z >z)$ for $z\in (-5,5)$ where $(X_1,X_2,X_3,Z)$ is a Gaussian vector such that marginals have unit variance. Larger dots correspond to larger values of $z$. %
  These figures are produced using simulations. %
  Computations throughout the article are made in {\tt R}~\cite{R}. Fig~1b was built using the package {\tt rgl}~\cite{rgl}.} %
\end{figure}

A more complex version of this problem has become central in finance since the 2007--2011 crisis. Let $X$ and $Y$ be two assets returns, and $Z$ be the overall volatility of the market. So as to quantify the risk a financial institution would face during a crisis, one must estimate $\rho_{XY\mid Z>z_c}$, for $z_c$ a crisis volatility threshold potentially larger than all observed $Z$~\cite{Campbell,Forbes,nature}. As Fig~\ref{fig:GCorrel} shows, $\rho_{XY\mid Z>z_c}$ can be markedly larger than $\rho_{XY\mid Z<z}$ even when $(X,Y,Z)$ is a Gaussian vector. Furthermore, Fig~\ref{fig:PCA}, shows that in higher dimensions, conditioning by $\{Z>z_c\}$ non-trivially affects the eigenvectors of the covariance matrix. It follows that efficient estimators of $\rho_{XY\mid Z>z_c}$ are needed by banks: to properly determine how much fund to set aside in provision of a crisis~\cite{Packham}, and to efficiently allocate assets during a crisis~\cite{Kenett2014}.

We present two estimators that address under minimal assumptions these problems. First, in Theorem~\ref{Thm1}, we propose an admissible estimator of $\rho_{XY\mid \A}$ for any $\A$ such that $\mathbb{P}\left(\A\right) > 0$. Second, in Theorem~\ref{Col2}, we present an estimator of $\rho_{XY}$ relying on a sample where for all realizations an event $\A$ is verified (henceforth referred to as an $\A$-sample). Using both estimators allows to estimate $\rho_{XY\mid \A}$ given a $\A'$-sample, this for any two events $\A$ and $\A'$ with non zero probabilities of occurring.

These results describe a highly counter-intuitive and non-trivial phenomena. As shown in Figs~\ref{fig:GCorrel} and \ref{fig:PCA}, event conditional correlation has a strikingly far from linear behavior even in the Gaussian case. This underlines how non-linear linear dependence can be. As the analysis will show, it is the homogeneous nature of correlation that allows to transport estimates under one condition to another. However, the scale at which it is observed under a given condition is driven by the conditional variances of the variables under $\A$.
 
The proposed estimators can directly be used to complement many statistical approaches, either to address sampling problems or to extend them to non-linear cases. We present three examples: the first focuses on a non-linear regression method that uses event conditional correlation (Section~\ref{seg-reg}); the second considers the power of event conditional correlation to test for independence while relying on a partial sample (Section~\ref{testing}); the third contrasts the realized and counterfactual topologies of a financial market across risk regimes (Section~\ref{contagion}.) We conclude on the implications of Theorems~\ref{Thm1} and~\ref{Col2} on the robustness to sampling of the leading eigenvectors of covariance matrices.

\section{Relations with other dependence parameters}%%%

Here we discuss how event conditional correlation generalizes many partial dependence parameters. We derive the properties of event conditional correlation starting from Section~\ref{ECC}.

We begin by formally defining $\rho_{XY\mid\A}$. For two real valued random variables $X$ and $Y$ defined on the probability space $(\Omega,\mathcal{F},\mathbb{P})$ and assuming $\A$ to be in $\mathcal{F}$ and such that $\mathbb{P}(\A)>0$, event conditional correlation verifies:
\begin{equation*}
\rho_{XY\mid \A} =\dfrac{
\E\left[
	\left(X-\E\left[X\vert\A\right]\right)
	\left(Y-\E\left[Y\vert\A\right]\right)
	\Big\vert\A\right]}
{\sqrt{\E\left[\left(X-\E\left[X\vert\A\right]\right)^2\Big\vert\A\right]
	 \E\left[\left(Y-\E\left[Y\vert\A\right]\right)^2\Big\vert\A\right]}}\cdot
\end{equation*}
Practically, as in the examples above, we will access $\A$ though a third random variable $Z$ defined on the same probability space as $X$ and $Y$. Then, $\A$ will take the form $\A = Z^{-1}(B)$, for $B$ a subset of the support of $Z$.

Event conditional correlation belongs to the class of conditional dependence parameters. All such parameters are built by considering the dependence between two variables $X$ and $Y$ while controlling for the behavior of a third variable $Z$. We will see that using events $\A$ allows us to replicate most such controls.

The simplest way of controlling $Z$ is to fix it at some value $z$. This is how conditional correlation, which we will write $\rho_{XY\mid Z=z}$, is built. The relationship between $\rho_{XY\mid Z=z}$ and $\rho_{XY\mid\A}$ can be formalized by considering a sequence $\A_t$ of events tending to $Z^{-1}(\{z\})$ and such that $\mathbb{P}(\A_t)>0$ for all $t$. If such a sequence exists, then $\smash{\rho_{XY\mid\A_t}}\to\smash{\rho_{XY\mid Z=z}}$. Fixing $Z$ at a given value is also how conditional copulas are built~\cite{acar2013,poczos2012copula,gijbels2012}. However, the drawback of fixing $Z$, is that unless $\mathbb{P}(Z=z)>0$ both conditional correlation and conditional copulas are only identifiable under parametric assumptions. On the other hand, we can always estimate event conditional correlation non-parametrically. Many other dependence parameters are built as the expectation over $Z$ of a conditional dependence parameter defined for a fixed $z$; e.g., liquid association~\cite{li2002}, incomplete Lancaster interaction~\cite{sejdinovic2013}, partial martingale difference correlation~\cite{park2015} and conditional information~\cite{poczos2012}. While taking the expectation makes it possible to use non-parametric methods in these cases, the obtained parameters lose their informational content regarding the local variations of the dependence.

An other way of controlling $Z$ is to push it into one of its tails. Such conditional dependence parameter take the form of limit event conditional correlation: $\lim_{z\to \infty}\rho_{XY\mid Z\geq z}$. This type of dependence parameter is explored further in~\cite{akemann1984}. Importantly, the case $Z = \min(X,Y)$ leads to the tail dependence parameter.

The influence of $Z$ over $(X,Y)$ can also be controlled by removing its linear effects on $X$ and $Y$. This leads to partial correlation---that we denote $\rho_{XY\mid Z}$. There also exists a direct relation between $\rho_{XY\mid\A}$ and $\rho_{XY\mid Z}$. We formalize it below in~\eqref{eq:Thm1Simple}.

Thus, $\rho_{XY\mid\A}$ is a bridge between partial, conditional and tail correlations. Furthermore, it can be used to describe when all or some of these parameters match in the neighborhood of a given $z$. This part of our work completes the discussion started in~\cite{Lawrance1976} and continued more recently in~\cite{Ritei,Baba2005}. However, there are no direct connections between event conditional correlation and correlation distance as introduced in~\cite{szŽkely2007}. Correlation distance and related dependence parameters aim to generalize correlation to test for dependence between random vectors. On the other hand, event conditional correlation aims to describe in more details the dependence between two scalar variables. Nonetheless, we show in Section~\ref{testing} that event conditional correlation can detect dependence when more complex measures do not, while in Section~\ref{Conc} we discuss the consequences of Theorem~\ref{Thm1} and~\ref{Col2} on the structure of conditional covariance matrices.

\section{New event conditional correlation estimator}\label{ECC}%%%%
In this section we present an admissible estimator of event conditional correlation that uses all the available sample rather than the subsample where the condition is verified, as is common practice. We will proceed in two steps, first introducing the hypotheses and notation, before presenting the general result along with a small-sample study (see Fig~\ref{MC_1}).

\begin{definition}
Let $X$ and $Y$ be two centered real valued univariate random variables with finite second moments. Let $Z_1$ and $Z_2$ be two real valued random vectors of the same dimension, possibly containing (or equal) to $X$ or $Y$, both having second moments. Finally assume $\A$ to be $Z_1$ and $Z_2$-measurable and of non-zero probability of occurring; i.e., there exists $B_1\subset \mathbb{R}^{dim(Z_1)}$ and $B_2\subset \mathbb{R}^{dim(Z_2)}$ such that $\A = Z_1^{-1}(B_1)\cap Z_2^{-1}(B_2)$ and $\mathbb{P}(\A)>0$. 

We use classical notation: $\rho$-s are the correlations and $\Sigma$-s are the covariance matrices of the variables in index (we use $\sigma$-s for the standard deviation in the univariate case), finally $\beta$-s are regression parameters and $\epsilon$-s are the regression residuals. For instance for $Z$ a centered scalar random variable we have:
\begin{equation*}
X = Z\beta_{XZ} + \epsilon_{XZ},
\end{equation*}
with $\beta_{XZ}$ the classical ordinary least squares regression parameter, equal to $\rho_{XZ}\sigma_X/\sigma_Z$.
\end{definition}

\begin{assumption}
We define here the two assumptions $\mathbf{A1}$ and $\mathbf{A2}$:
\begin{align*}
\mathbf{A1}(X,Y,Z_1,Z_2,\A):&\quad \Sigma_{\epsilon_{XZ_1}\epsilon_{YZ_2}}=\Sigma_{\epsilon_{XZ_1}\epsilon_{YZ_2}\mid \A},\\
\mathbf{A2}(X,Y,Z_1,Z_2,\A):&\quad \cov\left(Z_1\beta_{XZ_1},\epsilon_{YZ_2}\mid \A\right)+\cov\left(Z_2\beta_{YZ_2},\epsilon_{XZ_1}\mid \A\right)=0.
\end{align*}
\end{assumption}
These assumptions should be seen as minimal since $\mathbf{A1}$ is necessary according to~\cite{Ritei}, and if $Z_1 = Z_2$, $\mathbf{A2}$ is automatically verified. Meeting $\mathbf{A1}$ can be attained by adding the sufficient number of covariates in $Z_1$ and $Z_2$, something that was not possible before our contribution. Finally, if $\mathbf{A1}$ remains falsified, the bias induced in the following estimators can be controlled by $\|\Sigma_{Z_1,Z_2\mid\A}-\Sigma_{Z_1,Z_2}\|^{-2}$, making them still of interest in cases where this value is small.

\begin{theorem}\label{Thm1}%%%
Under $\mathbf{A1}$ and $\mathbf{A2}$, we have that:
\begin{equation}
\rho_{XY\mid \A} = \dfrac{\cov(X,Y)+\beta_{XZ_1}^{\top}\delta_\A(Z_1,Z_2)\beta_{YZ_2}}{\left[ \sigma_X^2+\beta_{XZ_1}^{\top}\delta_\A(Z_1,Z_1)\beta_{XZ_1}\right]^{\frac{1}{2}}\left[ \sigma_Y^2+\beta_{YZ_2}^{\top}\delta_\A(Z_2,Z_2)\beta_{YZ_2}\right]^{\frac{1}{2}}},
\label{eq:Thm1}
\end{equation}
with for all $i,j\leq2,\ \delta_\A(Z_i,Z_j) = \cov(Z_i,Z_j\mid \A) - \cov(Z_i,Z_j).$
\end{theorem}

\begin{proof}
To be found in Appendix~\ref{PrfThm1}.
\end{proof}

To obtain a better intuition of the result, we simplify the problem and assume that the variables are scaled and such that $Z_1=Z_2=Z$, with $Z$ univariate. Then~\eqref{eq:Thm1} becomes
\begin{equation}
\rho_{XY\mid \A}= \dfrac{\rho_{XY} + \rho_{XZ} \rho_{YZ}\delta}{\left[ 1 + \rho_{XZ}^2\delta\right]^{\frac{1}{2}} \left[ 1 + \rho_{YZ}^2\delta\right]^{\frac{1}{2}}},
\label{eq:Thm1Simple}
\end{equation}
with $\delta = \sigma_{Z\mid \A}^2/\sigma_Z^2-1$. This form shows that $\rho_{XY\mid \A}$ is driven by the conditional variance, and more precisely by $\delta$, the normalized shift in conditional variance between inside and outside of $\A$. In the limit case where $\mathbb{P}\left(\A\right) = 0$, we recover the recursive equation to compute partial correlation $\rho_{XY\mid Z}$, allowing us to relate the two dependence parameters.

\begin{figure}[t]
{\fontsize{6.1pt}{.7em}\selectfont
\hfill
\begin{minipage}[t]{0.33\textwidth}
	$\qquad$Normal(0,$\eta$):~$\theta=$~(0$\cdot$2,0$\cdot$4,0$\cdot$6,1)$\vphantom{\chi^2_\eta}$
	\includegraphics[width=\textwidth,height=\textwidth]{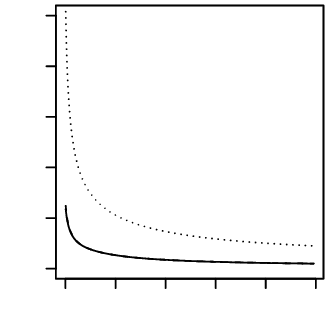}\vspace{-1.5\baselineskip}
	
	\noindent$\qquad$Normal(0,$\eta$):~$\theta=$~(0$\cdot$6,0$\cdot$7,0$\cdot$8,1e3)$\vphantom{\chi^2_\eta}$
       \includegraphics[width=\textwidth,height=\textwidth]{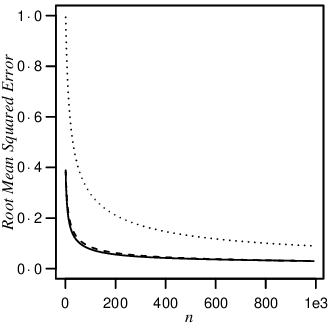}
\end{minipage}\hspace{-.25cm}
\begin{minipage}[t]{0.33\textwidth}
	$\qquad$Student-$t_\eta$:~$\theta=$~(0$\cdot$4,0$\cdot$5,0$\cdot$6,5)$\vphantom{\chi^2_\eta}$
	\includegraphics[width=\textwidth,height=\textwidth]{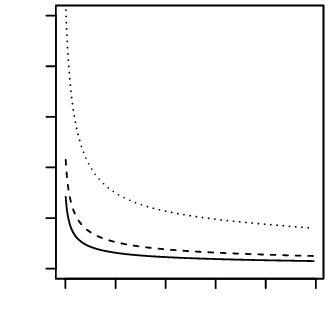}\vspace{-1.5\baselineskip}
	
	\noindent $\qquad$Student-$t_\eta$:~$\theta\!=\!$~(0$\cdot$2,0$\cdot$3,0$\cdot$4,30)$\vphantom{\chi^2_\eta}$
	\includegraphics[width=\textwidth,height=\textwidth]{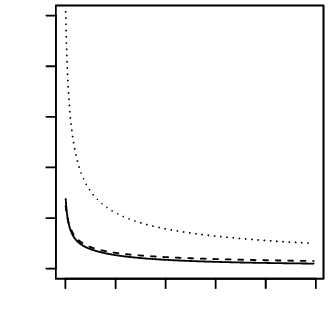}
\end{minipage}\hspace{-.25cm}
\begin{minipage}[t]{0.33\textwidth}
	$\quad$Normal(0,$\chi^2_\eta$):~$\theta=$~(0$\cdot$2,0$\cdot$6,0$\cdot$6,10)$\vphantom{\chi^2_\eta}$
	\includegraphics[width=\textwidth,height=\textwidth]{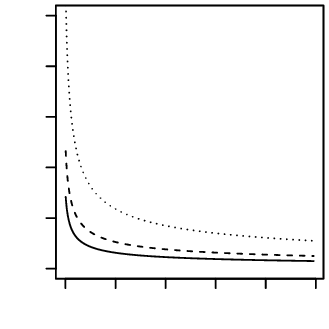}\vspace{-1.5\baselineskip}
	
	\noindent $\quad$ Normal(0,$\chi^2_\eta$):~$\theta=$~(0$\cdot$2,0$\cdot$3,0$\cdot$4,1)$\vphantom{\chi^2_\eta}$
	\includegraphics[width=\textwidth,height=\textwidth]{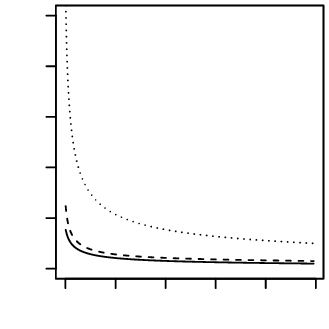}
\end{minipage}
\hfill\hfill}
	\caption{Root mean squared error of an event conditional correlation curve estimate, as shown in Fig~\ref{fig:GCorrel}, to its true value for different sample sizes $n$ for: the proposed method (dashed line), using sub-sampling (dotted line) and using well specified maximum likelihood initiated at the true value (solid line); each point is produced using 1000 simulations. We consider three distributions, each with two sets of parameters $\theta = \left(\rho_{XY},\rho_{XZ},\rho_{YZ},\eta\right)$.}
  \label{MC_1}
\end{figure}

Finally, \eqref{eq:Thm1Simple} recovers the results of~\cite{Boyer, Ladoucette, Forbes, Packham}, connecting our result with theirs. However, because all these works focus on risk measures in finance, they make field specific assumptions on the nature of the condition $\A$ and on the distribution of $(X,Y,Z)$, while we work under minimal assumptions.

We now draw from Theorem~\ref{Thm1} a new estimator of $\rho_{XY\mid \A}$. In the following we denote using hat estimators: for instance $\smash{\widehat \beta_{XY}}$ is an estimator of $\beta_{XY}$.

\begin{corollary}\label{Col1}%%%
Under $\mathbf{A1}$, $\mathbf{A2}$ and the assuming that $\widehat \cov(X,Y)$, $\widehat \sigma_X$, $\widehat \sigma_Y$, $\widehat\beta_{XZ_1}$, $\widehat\beta_{YZ_2}$, $\widehat{\delta_\A}(Z_1,Z_1)$, $\widehat{\delta_\A}(Z_2,Z_2)$ and $\widehat{\delta_\A}(Z_1,Z_2)$ are $\sqrt{n}$-consistent, asymptotically normal estimators, we have that
\begin{equation}
\dfrac{ \widehat\cov(X,Y)+\widehat\beta_{XZ_1}^{\top}\widehat{\delta_\A}(Z_1,Z_2)\widehat\beta_{YZ_2}}{\left[ \widehat \sigma_X^2+\widehat\beta_{XZ_1}^{\top}\widehat{\delta_\A}(Z_1,Z_1)\widehat\beta_{XZ_1}\right]^{\frac{1}{2}}\left[ \widehat \sigma_Y^2+\widehat\beta_{YZ_2}^{\top}\widehat{\delta_\A}(Z_2,Z_2)\widehat\beta_{YZ_2}\right]^{\frac{1}{2}}}
\label{eq:Cor1}
\end{equation}
is a $\sqrt{n}$-consistent, asymptotically normal, estimator of $\rho_{XY\mid \A}$.
\end{corollary}

\begin{proof}
The result is a direct application of Theorem~\ref{Thm1} and the delta method. We detail the proof in Appendix~\ref{PrfCol1}. Figure~\ref{MC_1} shows a small-sample study corroborating this result. To obtain an $\sqrt{n}$-consistent estimator of the variances---the $\smash{\widehat{\delta_\A}}$-s---we estimate the joint distribution of $(Z_1,Z_2)$ using the whole sample and infer from that estimate the conditional covariance matrices. In all six cases our estimator almost realizes the Cram\'er-Rao bound: on average it only adds a 0$\cdot$04 error in correlation estimates compared to the Cram\'er-Rao case.
\end{proof}

It is of interest to tell whether two event conditional correlation estimates are significantly different or not. To this end we must produce confidence intervals for our estimator. Since Corollary~\ref{Col1} is derived through the delta method, the variance stabilizing transformation, or inverse delta method, should be used~\cite{Fisher, Hotelling}. However, this method is not applicable here as it cannot be computed in closed form~\cite{vanderVaart}. Nevertheless, resampling methods can be used, which we recommend.
 
\section{Implied unconditional correlation estimator}\label{UCC}%%%
In this section we present an estimator of unconditional correlation based on an $\A$-sample (a sample where $\A$ is verified for all observations.) The formulation of this estimator is not intuitive, but as detailed below it is in fact driven by the same adjustment for conditional variance shift as~\eqref{eq:Thm1}. We are not aware of any other estimator to compare our result with, but in Fig~\ref{MC_2} we present a small-sample study where our estimator almost realizes the Cram\'er-Rao bound: On average it only adds a 0$\cdot$02 relative error in correlation estimates compared to the Cram\'er-Rao case.

\begin{theorem}\label{Col2}%%%
Assuming $\mathbf{A1}$, $\mathbf{A2}$ we have that
\begin{multline}
\rho_{XY} = \rho_{XY\mid \A} \left(1+R_{XZ_1}^{\top}\bar\delta_\A\left(Z_1,Z_1\right)R_{XZ_1}\right)^{\frac{1}{2}}\left(1+R_{YZ_2}^{\top}\bar\delta_\A\left(Z_2,Z_2\right)R_{YZ_2}\right)^{\frac{1}{2}}\\
-R_{XZ_1}^{\top}\bar\delta_\A\left(Z_1,Z_2\right)R_{YZ_2},
\label{eq:Col2}
\end{multline}
with for all $i,j\leq2,\ \bar\delta_\A(Z_i,Z_j)=diag(\Sigma_{Z_i}^{\ })\Sigma_{Z_i}^{-1}\delta_\A(Z_i,Z_j)\Sigma_{Z_j}^{-1}diag(\Sigma_{Z_j}^{\ })$ and 
\begin{equation*}
\begin{cases}
R_{XZ_1} 
&= \left\{\rho_{XZ_{1i}\mid\A}\left[1+\left(\sigma_{Z_{1i}\mid\A}^2\sigma_{Z_{1i}}^{-2}-1\right)\left(1-\rho_{XZ_{1i}\mid\A}^{2}\right)\right]^{-\frac12}\right\}_{i\leq dim(Z_1)}\\
R_{YZ_2} 
&= \left\{\rho_{YZ_{2j}\mid\A}\left[1+\left(\sigma_{Z_{2j}\mid\A}^2\sigma_{Z_{2j}}^{-2}-1\right)\left(1-\rho_{YZ_{2j}\mid\A}^{2}\right)\right]^{-\frac12}\right\}_{j\leq dim(Z_2)}.
\end{cases}
\end{equation*}
\end{theorem}

\begin{proof}
To be found in Appendix~\ref{PrfCol2}. The proof consists in two steps: i) inverting~\eqref{eq:Thm1} gives~\eqref{eq:Col2}, ii) inverting it again to compute the entries of $\beta_{XZ_1}$ and  $\beta_{YZ_2}$ yields the expressions for $R_{XZ_1}$ and $R_{YZ_2}$, which gives the result.
\end{proof}

Equation~\ref{eq:Col2} does not make the link between conditional variance and event conditional correlation explicit. Rewriting it in the simplified case of~\eqref{eq:Thm1Simple}, shows that it is in fact based on exactly the same transformation as that in~\eqref{eq:Thm1}: with $\bar\delta = \sigma_{Z}^2/\sigma_{Z\mid \A}^2-1$, we have 
\begin{equation*}
\rho_{XY}=
	\dfrac{\rho_{XY\mid \A}+\rho_{XZ\mid \A}\rho_{YZ\mid \A}\bar\delta}{\left[1+\rho_{XZ\mid \A}^2\bar\delta\right]^{\frac{1}{2}}\left[1+\rho_{YZ\mid \A}^2\bar\delta\right]^{\frac{1}{2}}}.
\end{equation*}
In fact, in that specific setting, for any event $\A'$ and with $\tilde\delta = \sigma_{Z\mid\A'}^2/\sigma_{Z\mid \A}^2-1$, we have 
\begin{equation*}
\rho_{XY\mid\A'}=
	\dfrac{\rho_{XY\mid \A}+\rho_{XZ\mid \A}\rho_{YZ\mid \A}\tilde\delta}{\left[1+\rho_{XZ\mid \A}^2\tilde\delta\right]^{\frac{1}{2}}\left[1+\rho_{YZ\mid \A}^2\tilde\delta\right]^{\frac{1}{2}}}.
\end{equation*}

\begin{figure}[t]
{\fontsize{6.1pt}{.7em}\selectfont
\hfill
\begin{minipage}[t]{0.33\textwidth}
	$\qquad$Normal(0,$\eta$):~$\theta=$~(0$\cdot$2,0$\cdot$4,0$\cdot$6,1)$\vphantom{\chi^2_\eta}$
	\includegraphics[width=\textwidth,height=\textwidth]{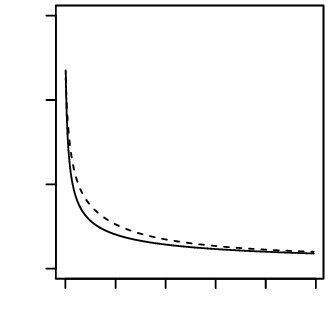}\vspace{-1.5\baselineskip}

	\noindent $\qquad$Normal(0,$\eta$):~$\theta=$~(0$\cdot$6,0$\cdot$7,0$\cdot$8,1e3)$\vphantom{\chi^2_\eta}$
	\includegraphics[width=\textwidth,height=\textwidth]{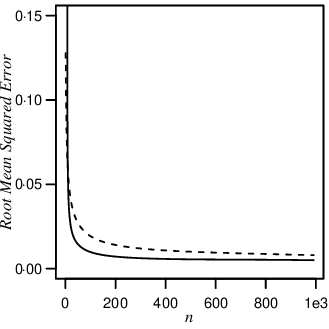}
\end{minipage}\hspace{-.25cm}
\begin{minipage}[t]{0.33\textwidth}
	$\qquad$Student-$t_\eta$:~$\theta=$~(0$\cdot$4,0$\cdot$5,0$\cdot$6,5)$\vphantom{\chi^2_\eta}$
	\includegraphics[width=\textwidth,height=\textwidth]{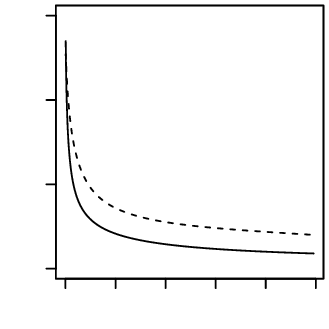}\vspace{-1.5\baselineskip}

	\noindent $\qquad$Student-$t_\eta$:~$\theta\!=\!$~(0$\cdot$2,0$\cdot$3,0$\cdot$4,30)$\vphantom{\chi^2_\eta}$
	\includegraphics[width=\textwidth,height=\textwidth]{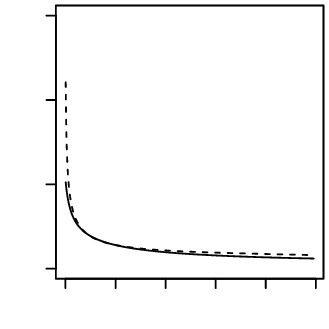}
\end{minipage}\hspace{-.25cm}
\begin{minipage}[t]{0.33\textwidth}
	$\quad$Normal(0,$\chi^2_\eta$):~$\theta=$~(0$\cdot$2,0$\cdot$6,0$\cdot$6,10)$\vphantom{\chi^2_\eta}$
	\includegraphics[width=\textwidth,height=\textwidth]{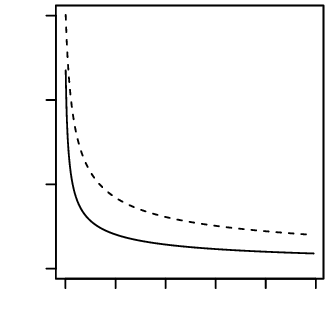}\vspace{-1.5\baselineskip}

	\noindent $\quad$Normal(0,$\chi^2_\eta$):~$\theta=$~(0$\cdot$2,0$\cdot$3,0$\cdot$4,1)$\vphantom{\chi^2_\eta}$
	\includegraphics[width=\textwidth,height=\textwidth]{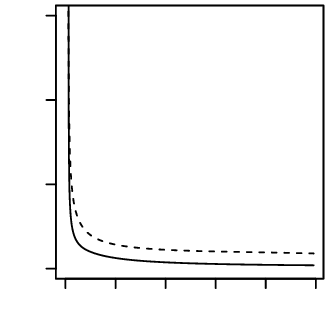}
\end{minipage}
\hfill\hfill}
	\caption{Root mean squared error of ten estimates based on sub-samples constructed as in Fig~\ref{fig:GCorrel} to the true value for different sample sizes $n$ for: the proposed method (dashed line) and using well specified maximum likelihood initiated at the true value (solid line). Otherwise this small-sample study is structured as that of Fig~\ref{MC_1}.}
	\label{MC_2}
\end{figure}

From Theorem~\ref{Col2} we naturally obtain an estimator of the unconditional correlation:
\begin{corollary}\label{Col3}%%%
With a $\A$-sample, under $\mathbf{A1}$, $\mathbf{A2}$ and the additional assumption that $\widehat \rho_{XY\mid \A}$, $\widehat\beta_{XZ_1\mid \A}$, $\widehat\beta_{YZ_2\mid \A}$, $\widehat{\delta_\A}(Z_1,Z_1)$, $\widehat{\delta_\A}(Z_2,Z_2)$ and $\widehat{\delta_\A}(Z_1,Z_2)$ are $\sqrt{n}$-consistent, asymptotically normal estimators, using~\eqref{eq:Col2} yields a $\sqrt{n}$-consistent, asymptotically normal estimators of $\rho_{XY}$.
\end{corollary}

\begin{proof}
The result is a direct consequence of Theorem~\ref{Col2} and the delta method. A proof is presented in Appendix~\ref{PrfCol3} and small sample simulations are presented in Fig~\ref{MC_2}. To obtain $\sqrt{n}$-consistent estimators of the conditional covariances, the $\smash{\widehat{\delta_\A}}$-s, we use maximum likelihood with a truncated distribution on $(Z_1,Z_2)$ and infer from the estimate the corresponding covariance.
\end{proof}

\section{Examples}%%%%%%%%%%%%%%
We now use the tools developed in Theorems~\ref{Thm1} and ~\ref{Col2} to address three problems: non-linear regression, dependence testing, and financial dependence structure. In the first two examples we use synthetic datasets, and in the last one we will consider the NASDAQ-100 index within and without the recent financial crisis.

\subsection{Piecewise affine functional regression}\label{seg-reg}
\begin{figure}
  \hfill
  \begin{minipage}[t]{0.45\textwidth}
       \includegraphics[width=\textwidth,height=\textwidth]{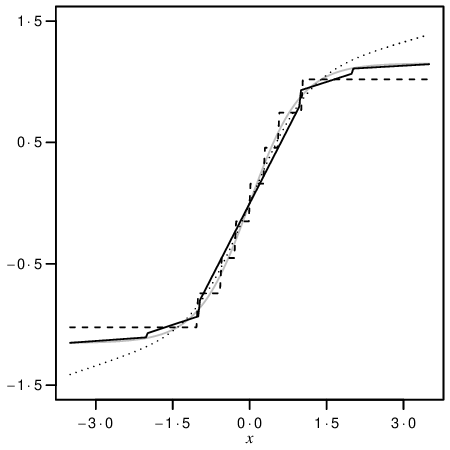}
  \end{minipage}
  \begin{minipage}[t]{0.45\textwidth}
       \includegraphics[width=\textwidth,height=\textwidth]{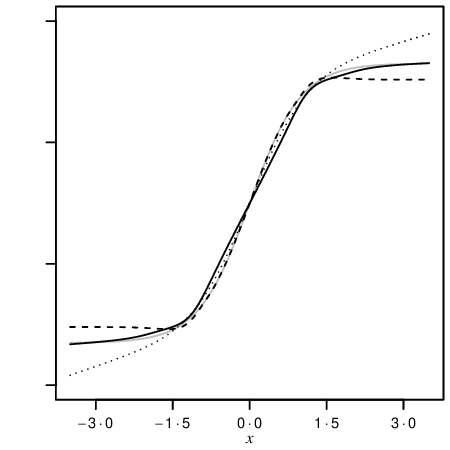}
  \end{minipage}
  \hfill\hfill
  \caption{Comparison of three regression methodologies: generalized additive regression (dotted line; produced with the package {\tt gam}~\cite{gamR}), tree regression (dashed lines; produced with the package {\tt tree}~\cite{tree}), and one produced using event conditional correlation (solid line, corresponding to $\smash{\hat f}$ in \eqref{eq:ecc_fun_reg}). On the left are the raw function estimates and on the right we present spline smoothed versions of the same estimates for fair comparison. In both cases, the true response function is plotted in continuous grey.}
  \label{Example}
\end{figure}

We compare three regression methodologies: generalized additive regression~\cite{GAM}, tree regression~\cite{CART}, and a new regression methodology we introduce that uses event conditional correlation. The synthetic dataset used consists in 1e4 realizations of two variables $X$ and $Y$ such that $Y$ is equal to $\tanh(X)+\epsilon$, with $\epsilon$ a centered Gaussian noise.

The regression estimate that uses event conditional correlation is obtained in three steps and builds on the concept of segmented regression~\cite{liu1997segmented}. First, we break the support of $X$ into several disjoint intervals of the same size, say the $\{A_i\}_i$. Then, using Corollary~\ref{Col1}, for each $A_i$ we estimate the correlation between $Y$ and $X$ conditionally on $X$ being in $A_i$ and obtain the corresponding regression slope $\hat\beta_i$. The final estimate is the piecewise affine function built using the regression slopes: 
\begin{equation}\label{eq:ecc_fun_reg}
\hat f: x\mapsto\sum_i1_{\{x\in A_i\}}[\hat m_i+\hat\beta_i (x-\hat m'_i)],
\end{equation}
where $(\hat m_i,\hat m_i')$ is the empirical mean of $(Y,X)$ conditionally on $X$ being in $A_i$.

We present the obtained estimates in Fig~\ref{Example} and observe that all three methods perform comparably well. However, we note that the estimate produced using conditional correlation is the only one to capture the tail of $Y$. This leads to a better root mean squared error (0$\cdot$051 compared to 0$\cdot$12 and 0$\cdot$081 for generalized additive regression and tree regression respectively).

\subsection{Test for bivariate dependence}\label{testing}
\begin{figure}
	\centering
	\includegraphics[width=.5\textwidth,height=.5\textwidth]{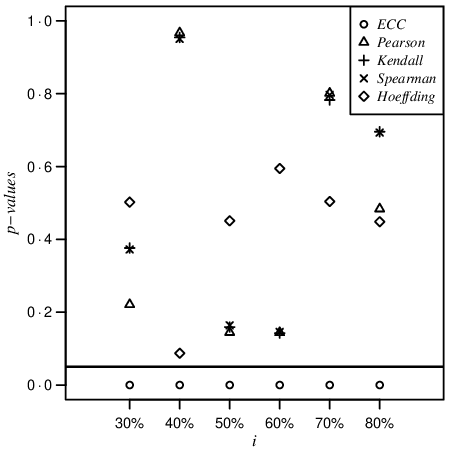}
  \caption{On the $x$-axis is $i$, the quantile considered, and on the $y$-axis is the $p$-value of the five considered tests indexed in the top right legend (ECC standing for event conditional correlation). The horizontal line is the 0\! $\cdot$05 threshold.}
\label{pvalues}
\end{figure}

We compare five bivariate tests for dependence between $X$ and $Y$ given a $\A$-sample. We find that event conditional correlation outperforms all the other methods in the considered synthetic examples.

The tests considered are as follows. The first four test for a dependence parameter being equal to zero. The considered parameters are: i) the implied unconditional correlation estimated using Corollary~\ref{Col3} assuming that the unconditional variance of the covariate is 1, ii) the Pearson, iii) Spearman, and iv) Kendall dependence parameters. The last considered test is the Hoeffding test, which is the univariate case of correlation distance~\cite{szŽkely2007}.

The total synthetic dataset consists of 5e3 realizations of $(X,Y,Z)$, a trivariate Gaussian vector so that all entries have unit variance and that $\rho_{XY}=-$0$\cdot$25 and $\rho_{YZ} = -\rho_{XZ}=$0$\cdot$50. We then apply the five considered tests to the $\A$-samples of $(X,Y)$ where $\A$ takes the same form as in Figure~\ref{fig:GCorrel}:
\[\A = \{Z\in [Q_Z(i-0\cdot1),Q_Z(i)]\},\] 
for $i\in\{$0$\cdot$1,0$\cdot$2,$\dots$,1$\}$. Then, each $\A$-sample consists of 5e2 realizations.

Only the test based on our results succeeds in detecting the dependence between $X$ and $Y$ (see Fig~\ref{pvalues}). Because of the sampling constraint, the co-movement of $X$ and $Y$ are limited, making classical tests unable to detect the dependence. The estimator of Corollary~\ref{Col3} allows to magnify this dependence, and hence detect it.

\subsection{Financial dependence network}\label{contagion}
We now turn to one of the original objectives of our study of event conditional correlation: financial assets dependence. One prominent question since the contribution of~\cite{Forbes} is the following: Do crises alter the dependence structure of financial assets? By exhibiting a special case of Theorem~\ref{Thm1}, ~\cite{Forbes} argued that it does not. This result raised a long and involved controversy that we do not review here.

Our new results allows us to consider the same problem in a more general setting: i) we can use varied covariates to describe the shift in dependence structure, whereas~\cite{Forbes} considered dependence at the bivariate level, of the form $\rho_{XY\mid Y>y}$ for some threshold $y$; ii) we can use the correlation structure as a whole, instead of considering each correlation parameters independently. We will do so for the NASDAQ-100, a widely traded financial index representing more than 15\% of US equity market ({\tt nasdaqtrader.com}, May 2015).

\subsubsection{Dataset} We consider the daily closing quotes between 2007-07-23 and 2015-06-30 of all the NASDAQ-100 components at the last date, totaling 2e3 trading days. We only kept the components that were already listed on the first date of the dataset, keeping 94 out of 100\footnote{All quotes were obtained using the {\tt quantmod} package~\cite{quantmod}.}. For each component, we extract the normalized residuals of the log-returns\footnote{We do so using a auto regressive model for the log returns (AR model), and model residuals with the generalized auto-regressive conditional heteroskedasticity model (GARCH model). We use the {\tt fGarch} package~\cite{fGarch} to jointly estimate these models.}. We call $X$ the matrix containing all the obtained normalized residuals.

The covariate we use is the CBOE-NASDAQ-100 volatility index over the same period; we call it $Z$. This index describes at each date the level of risk of the NASDAQ-100. To obtain a precise description of the dependence, we will use $Z$, $Z^2$, $Z^3$ and the lag of these three quantities\footnote{We selected these lags and powers using classical analysis of variance methods.}. We call $W$ the full matrix of regressors.

\subsubsection{Correcting the correlations} 
We define two regimes: the {\em crisis} regime, where $Z$ is in it's higher quartile, and the complementary, the {\em stable} regime. Changing from the higher quartile to an other threshold does not substantially affect the results.

Using each subsample independently, we can estimate the correlation matrices of $X$ inside and outside of the crisis regime. If we do so, we find that more than 97\% of correlation coefficients are significantly different across regimes\footnote{We use resampling methods and the $t$-test. In this case, and throughout the remainder of this analysis, we use Bonferroni correction and 0$\cdot$1\% significance levels.}.

Using Theorem~\ref{Col2}, we can correct for the bias induced by sampling: We use both sample independently to produce estimate of the unconditional correlations, the correction being made using $W$ as covariate. After correction, we find that less than 83\% of the correlations are significantly different.

Thus, as in~\cite{Forbes}, we observe that ignoring subsampling inflates the apparent shift in dependence structure across regimes. However, as opposed to~\cite{Forbes}, there remains an important shift between regimes. Using Theorem~\ref{Thm1} will allow us to describe this shift in dependence structure.

\subsubsection{Describing the contagion}
We now consider the network linking the components of the NASDAQ-100. Considering this network allows to describe dependence as a whole and does not require multiple testing. The nodes of the network are the components of the NASDAQ-100. The weight of the edge between two nodes is the partial correlation between the two components: the correlation after removing linear effects from all the other components.

We first compare the networks linking NASDAQ-100 components within and without the crisis regime, this is done while correcting for sampling using Theorem~\ref{Col2}. To compare these networks we will use the nodes' centrality scores. (These quantities describe the importance of each node in the network in terms of how likely a random walk over the network is to visit that node~\cite{newman2010networks}.) More precisely, we will compare the sample average and standard deviation of the centrality scores of all nodes in the network\footnote{These moments are tied to the first eigenvector of the covariance matrix of $X$~\cite{newman2010networks}. Thus, by continuity arguments and using the Delta method, these moment estimates are expected to be asymptotically normal. We used the Shapiro-Normality-Test and failed to reject that the bootstrap distribution is normal in all cases.}.

\begin{figure}[t]
	\includegraphics{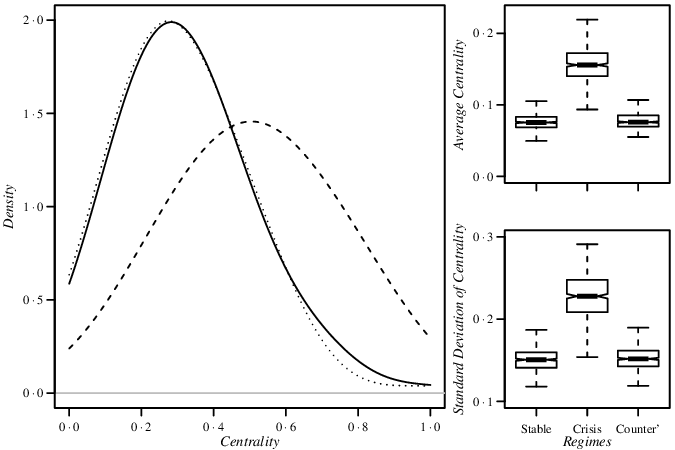}
	\caption{Centrality of NASDAQ across regimes. On the left is the empirical density of the NASDAQ components' centrality scores across regimes (stable, crisis and counterfactual regimes in solid, dashed and dotted lines respectively.) On the right is the box-plot of bootstrap samples of the average and standard deviation of the networks centrality under the three regimes: stable, crisis and counterfactual (denoted counter').}
\label{nasdaq}
\end{figure}

We find that both the sample averages and standard deviations of centrality show marked upward shifts between the stable and crisis regimes (see Fig~\ref{nasdaq}.) Further, since the network estimates are already corrected for sampling and the level of risk, the observed shifts must be caused by some other phenomena.

We explore this idea by evaluating the average and standard deviation of centrality in a counterfactual network. To build the counterfactual network, we use the realization from the stable regime along with Theorems~\ref{Thm1} and~\ref{Col2} to produce an estimate of the network we would observe if the variance of $Z$ was increased by $\delta(\sigma_{Z\mid \text{crisis}}^2-\sigma_{Z\mid \text{stable}}^2)$. In Fig~\ref{nasdaq} we set $\delta=5$ and observe that the counterfactual network does not replicate the centrality distribution observed during the crisis regime. Other values of $\delta$ do not affect this observation.

To conclude, we showed that: i) the observed change in dependence structure between the stable and crisis regimes must be caused by contagion and that ii) the structure observed during the crisis regime cannot be replicated by stressing the dependence structure of the stable regime. Interestingly, between the stable and crisis regimes, the network is moving from a state where most nodes have small centrality (except a select few), to a state where centrality is on average higher, but otherwise much more spread out. Then, although a few leading components may drive the behavior of the NASDAQ during the stable regime, this does not appear to be the case during the crisis regime. We expect that this effect, beyond the increase in variance, makes the market much harder to predict as it presents no clear leading variable.

\section{Discussion}\label{Conc}%%%%%%%%%%%%%%
Our results show that what drives correlation variations across events is the shift in conditional variances across events. It follows that by estimating the conditional variance shift, or by assuming a given value for it, we can estimate and compare correlation conditionally to any events using any partial sample. This is done while making no assumption on the dependence occurring between $X$ and $Y$. The only requirement of the method is to possess covariates able to describe the said dependence across conditions (the $(Z_1,Z_2)$.) Furthermore, the proposed estimators are consistent, asymptotically normal and display good small sample properties (see Figs~\ref{MC_1} and~\ref{MC_2}.)

These results have direct methodological applications. We provide three examples: one for non-linear regression, one for dependence testing, and one for financial networks. Event conditional correlation based approaches prove more powerful than comparable methods in the first two examples. In the last example, we characterized the occurrence of contagion (structural shift in the dependence structure) during the 2007--2011 financial crisis. Furthermore, we qualitatively described the uncovered crisis regime using counterfactuals.

Importantly, our last example exhibits a hidden consequence of Theorems~\ref{Thm1} and~\ref{Col2}: that the leading eigenvalues and eigenvectors of covariance matrices are robust to partial samples. To see this, consider a random vectors $X$ and a covariate $Z$. A direct consequence of Theorem~\ref{Thm1} is that $\Delta = \Sigma_X-\Sigma_{X\mid\A}$ is of rank at most $dim(Z)$. Then, for small $dim(Z)/dim(X)$, the effect of conditioning by $\A$ is very limited on the eigenvalues and eigenvectors of the sample covariance matrix. However, if the dimensions of $X$ and $Z$ are comparable, the effects cannot be neglected. We provide an example of this later point in dimension 3 in Fig~\ref{fig:PCA}.

An interesting final case to consider is when the dimension of $X$ goes to infinity~\cite{gao2015scca}. Because of its lower dimensionality, $\Delta$ will be much less affected than $\Sigma_X$ by this high dimensional setting. We conjecture that in this setting spectral methods become asymptotically frail to conditioning, even when the dimension of $Z$ remains fixed.

\appendix
\section*{Appendix}%%%%%%%%%%
\begin{proof}[Proof of Theorem~\ref{Thm1}]%%%
\label{PrfThm1}
We will proceed in two steps: first we will compute the covariance of $X$ and $Y$ knowing $\A$ and then proceed to compute the variance of $X$ and $Y$ knowing $\A$. In the following we will put $\A$ as index to operators used conditionally to $\A$: for instance $\E_{\A}\left[X\right]=\E\left[X\mid\A\right]$.

\paragraph{Covariance}
\begin{equation*}
\begin{array}{l c l}
{\cov}_{\A}\left(X,Y\right)
&=& \E_{\A}\left[\left(X-\E_{\A}\left[X\right]\right)^{\top}\left(Y-\E_{\A}\left[Y\right]\right)\right]\\[10pt]
&=& \E_{\A}\left[\left( Z_1 \beta_{XZ_1}+ \epsilon_{XZ_1}-\E_{\A}\left[ Z_1 \beta_{XZ_1} + \epsilon_{XZ_1}\right]\right)^{\top}\right.\\[10pt]
&&\quad\qquad\left.\left( Z_2 \beta_{YZ_2} + \epsilon_{YZ_2}-\E_{\A}\left[ Z_2 \beta_{YZ_2} + \epsilon_{YZ_2}\right]\right)^{\vphantom{P}} \right]\\[10pt]
&=& \beta_{XZ_1}^\top \E_{\A}\left[\left(Z_1-\E_{\A}\left[Z_1\right]\right)^\top\left(Z_2-\E_{\A}\left[Z_2\right]\right)\right]\beta_{YZ_2}\\[10pt]
&&\ \ \ +\  \E_{\A}\left[\left(Z_1-\E_{\A}\left[Z_1\right]\right)\left(\epsilon_{YZ_2}-\E_{\A}\left[\epsilon_{YZ_2}\right]\right)\right]\beta_{XZ_1}\\[10pt]
&&\ \ \ +\  \E_{\A}\left[\left(Z_2-\E_{\A}\left[Z_2\right]\right)\left(\epsilon_{XZ_1}-\E_{\A}\left[\epsilon_{XZ_1}\right]\right)\right]\beta_{YZ_2}\\[10pt]
&&\ \ \ +\ \E_{\A}\left[\left(\epsilon_{XZ_1}-\E_{\A}\left[\epsilon_{XZ_1}\right]\right)\left(\epsilon_{YZ_2}-\E_{\A}\left[\epsilon_{YZ_2}\right]\right)\right].
\end{array}
\end{equation*}
Under $\mathbf{A1}$ and $\mathbf{A2}$, we can simplify the above equation to:
\begin{equation}
{\cov}_{\A}\left(X,Y\right) =  \beta_{XZ_1}^\top {\cov}_{\A}\left(Z_1,Z_2\right)   \beta_{YZ_2}+ {\cov}(\epsilon_{XZ_1},\epsilon_{YZ_2}).
\label{cov1}
\end{equation}
Let use now compute ${\cov}(\epsilon_{XZ_1},\epsilon_{YZ_2})$. To do so we use the fact that~\eqref{cov1} is verified for any event $\A$ such that $\mathbf{A1}(X,Y,Z_1,Z_2,\A)$ and $\mathbf{A2}(X,Y,Z_1,Z_2,\A)$ are verified. This is the case if $\mathbb{P}\left(\A\right) = 1$, so that:
\begin{equation*}
{\cov}(X,Y) =  \beta_{XZ_1}^\top {\cov}\left(Z_1,Z_2\right)\beta_{YZ_2} + {\cov}(\epsilon_{XZ_1},\epsilon_{YZ_2}),
\end{equation*}
and we obtain:
\begin{equation}
{\cov}(\epsilon_{XZ_1},\epsilon_{YZ_2}) =  {\cov}(X,Y) - \beta_{XZ_1}^{\top} {\cov}\left(Z_1,Z_2\right)\beta_{YZ_2}.
\label{cov2}
\end{equation}
Hence, merging~\eqref{cov1} and~\eqref{cov2} we obtain:
\begin{equation}
{\cov}_{\A}\left(X,Y\right) =  {\cov}(X,Y) +  \beta_{XZ_1}^{\top} \delta_\A(Z_1,Z_2)  \beta_{YZ_2}.
\label{cov}
\end{equation}

\paragraph{Variance}
\begin{eqnarray}
{\var}_{\A}\left(X\right) &=&  {\var}_{\A}\left[ Z_1 \beta_{XZ_1}+\epsilon_{XZ_1}\right]\nonumber\\
&=&  \beta_{XZ_1}^\top {\var}_{\A}\left(Z_1\right)  \beta_{XZ_1}\ +  {\var}\left(\epsilon_{XZ_1}\right).
\label{var1}
\end{eqnarray}
We make the simplification using $\mathbf{A1}$. Let us now compute the variance of $\epsilon_{XZ_1}$. To do so we use the fact that~\eqref{var1} is verified for any event $\A$ such that $\mathbf{A1}(X,Y,Z_1,Z_2,\A)$ and $\mathbf{A2}(X,Y,Z_1,Z_2,\A)$ are verified. This is the case if $\mathbb{P}\left(\A\right) = 1$, then the above equation writes:
\begin{equation*}
{\var}(X) =   \beta_{XZ_1}^\top {\var}\left[Z_1\right]   \beta_{XZ_1} + {\var}\left(\epsilon_{XZ_1}\right),
\end{equation*}
so that:
\begin{equation}
{\var}\left(\epsilon_{XZ_1}\right) = {\var}\left(X\right) -  \beta_{XZ_1}^\top {\var}\left(Z_1\right)  \beta_{XZ_1}.
\label{var2}
\end{equation}
By merging~\eqref{var1} and~\eqref{var2}:
\begin{eqnarray}
{\var}_{\A}\left(X\right) &=&  {\var}\left(X\right) +  \beta_{XZ_1}^{\top}\delta_\A(Z_1,Z_1)  \beta_{XZ_1},\nonumber\\
{\var}_{\A}\left(Y\right) &=&  {\var}\left(Y\right) +  \beta_{YZ_2}^{\top}\delta_\A(Z_2,Z_2)  \beta_{XZ_2}.
\label{var}
\end{eqnarray}
The result for $Y$ is obtained similarly.

Merging the~\eqref{cov} and~\eqref{var} and using that both $X$ and $Y$ are normed we obtain~\eqref{eq:Thm1}.
\end{proof}

\begin{proof}[Proof of Corollary~\ref{Col1}]\label{PrfCol1}%%%
For simplicity we will work here in the simplified case presented after Theorem~\ref{Thm1}, where the variables are normed, $Z_1 = Z_2 = Z$ and $Z$ is univariate. The result and the proof extends directly to the general setting of Corollary~\ref{Col1}.

We call $\theta$ the vector composed of the parameters required to compute the $\rho_{XY\mid\A}$: $\theta = \left(\rho_{XY},\rho_{XZ},\rho_{YZ},\delta\right)$. We denote $\hat \theta$ an estimator of $\theta$ that converges at rate $n^{1/2}$ as in the statement of Corollary~\ref{Col1} and $\Sigma_{\theta}$ the asymptotic covariance matrix of $\theta$.

Let $\phi:R^4\rightarrow R$ be the map defined by:
\begin{equation*}
\phi(a,b,c,d)  = \dfrac{a + b c  d}{(1+b^2 d)^{\frac{1}{2}}(1+c^2 d)^{\frac{1}{2}}}.
\end{equation*}
The map $\phi$ is such that $\phi(\theta) = \rho_{XY\mid \A}$. We denote $\nabla \phi$ the gradient of $\phi$:
\begin{equation*}
\nabla \phi = \left(\dfrac{\partial \phi}{\partial a}, \dfrac{\partial \phi}{\partial b}, \dfrac{\partial \phi}{\partial c}, \dfrac{\partial \phi}{\partial d}\right).
\end{equation*}

Then under our assumptions we have that $\phi(\hat \theta)$ is an asymptotically normal estimator of $\phi(\theta)$ converging at the same rate $\nu$, and of asymptotic variance $\nabla \phi  \Sigma_{\theta}  \nabla \phi^{\top}$ as a direct application of the Delta Method, see~\cite[p. 30]{vanderVaart} (with the same notation).
\end{proof}

\begin{proof}[Proof of Theorem~\ref{Col2}]\label{PrfCol2}%%%
We start from~\ref{eq:Thm1}:
\begin{align}
\nonumber
\rho_{XY\mid \A}
&=\frac{\cov(X,Y)
+\beta_{XZ_1}^{\top}\delta_\A\left(Z_1,Z_2\right)\beta_{YZ_2}^{\ }}{
\left(\sigma_X^2+\beta_{XZ_1}^{\top}\delta_\A\left(Z_1,Z_1\right)\beta_{XZ_1}^{\ }\right)^{\frac{1}{2}}
\left(\sigma_Y^2+\beta_{YZ_2}^{\top}\delta_\A\left(Z_2,Z_2\right)\beta_{YZ_2}^{\ }\right)^{\frac{1}{2}}}\\
&=\frac{\rho_{XY}
+\frac{\beta_{XZ_1}^{\top}}{\sigma_X}\delta_\A\left(Z_1,Z_2\right)\frac{\beta_{YZ_2}^{\ }}{\sigma_Y}}{
\left(1+\frac{\beta_{XZ_1}^{\top}}{\sigma_X}\delta_\A\left(Z_1,Z_1\right)\frac{\beta_{XZ_1}^{\ }}{\sigma_X}\right)^{\frac{1}{2}}
\left(1+\frac{\beta_{YZ_2}^{\top}}{\sigma_Y}\delta_\A\left(Z_2,Z_2\right)\frac{\beta_{YZ_2}^{\ }}{\sigma_Y}\right)^{\frac{1}{2}}}.
\label{PrfCol2_1}
\end{align}
We then replace $\beta_{XZ_1}$ and $\beta_{YZ_2}$ by their value. To this en we introduce $\tilde R_{XZ_1} = \{\rho_{XZ_{1i}}\}_{i\leq dim(Z_1)}$ and $\tilde R_{YZ_2} = \{\rho_{YZ_{2j}}\}_{j\leq dim(Z_2)}$, the vectors of correlation between $X$ and $Y$ against $Z_1$ and $Z_2$ respectively, so that
\begin{equation*}
\beta_{XZ_1} = \tilde R_{XZ_1}\sigma_Xdiag(\Sigma_{Z_1})\Sigma_{Z_1}^{-1}
\quad\text{and}\quad
\beta_{YZ_2} = \tilde R_{YZ_2}\sigma_Ydiag(\Sigma_{Z_2})\Sigma_{Z_2}^{-1}.
\end{equation*}
From~\eqref{PrfCol2_1} we now recover
\begin{equation}
\rho_{XY\mid \A}
=\frac{\rho_{XY}
+\tilde R_{XZ_1}^{\top}\bar\delta_\A\left(Z_1,Z_2\right)\tilde R_{YZ_2}}{
\left(1+\tilde R_{XZ_1}^{\top}\bar\delta_\A\left(Z_1,Z_1\right)\tilde R_{XZ_1}^{\ }\right)^{\frac{1}{2}}
\left(1+\tilde R_{YZ_2}^{\top}\bar\delta_\A\left(Z_2,Z_2\right)\tilde R_{YZ_2}^{\ }\right)^{\frac{1}{2}}}\cdot
\end{equation}
It follows that it is sufficient to obtain the result to show that $R_{XZ_1}$ and $R_{YZ_2}$ are equal to $\tilde R_{XZ_1}$ and $\tilde R_{YZ_2}$ respectively. To do so we use Theorem~\ref{Thm1} to evaluate $\rho_{XZ_{1i}\mid\A}$ and obtain (as in the simplified case presented after the said theorem):
\begin{align*}
\rho_{XZ_{1i}\mid\A}
& = \frac{\rho_{XZ_{1i}} + \rho_{XZ_{1i}}\left(\frac{\sigma_{Z_{1i}\mid\A}^2}{\sigma_{Z_{1i}}^2}-1\right)}{\left(1+\rho^2_{XZ_{1i}}\left(\frac{\sigma_{Z_{1i}\mid\A}^2}{\sigma_{Z_{1i}}^2}-1\right)\right)^\frac12\left(1+\left(\frac{\sigma_{Z_{1i}\mid\A}^2}{\sigma_{Z_{1i}}^2}-1\right)\right)^\frac12},
\end{align*}
which concludes the proof.
\end{proof}

\begin{proof}[Proof of Corollary~\ref{Col3}]\label{PrfCol3}
In the same fashion as for Corollary 2 we will consider here only the simplified framework where $Z_1 = Z_2 = Z$ and with $Z$ univariate. The result and the proof extends to the multivariate setting.
Then the proof is the same as that of Corollary 2 using $\theta = \left(\rho_{XY},\rho_{XZ},\rho_{YZ},\bar \delta\right)$ with the exact same function $\phi$.
\end{proof}
%%%%%%%%%%%%%%%%%%%%%%%%%%%%%%%%%
\bibliographystyle{imsart-nameyear}%%%%%%%%%%%%%%%
\bibliography{ECC}%%%%%%%%%%%%%%%%%%%%%%%
%%%%%%%%%%%%%%%%%%%%%%%%%%%%%%%%%
\end{document}